\documentclass[conference, a4paper]{IEEEtran}           

\usepackage[utf8]{inputenc} 
\usepackage[T1]{fontenc}
\usepackage{url, balance}              

\usepackage[cmex10]{amsmath}  
\usepackage{amsfonts}
\usepackage{amssymb}

\usepackage{dsfont}
\newcommand{\mathbbm}[1]{\mathds{#1}}  

\usepackage[normalem]{ulem}
\usepackage{mathtools}
\usepackage{graphicx}
\usepackage{nicefrac}
\usepackage{bm}
\usepackage{xcolor, xpatch}
\usepackage[linesnumbered,ruled]{algorithm2e}
\usepackage{tikz}
\usetikzlibrary{positioning, arrows.meta}
\usepackage{enumitem}
\usepackage{booktabs}         
\usepackage[backend=bibtex,bibstyle=ieee,citestyle=numeric-comp,url=false,doi=false,isbn=false]{biblatex}
\usepackage[font=footnotesize,skip=10pt]{caption}

\usepackage{pgfplots}

\makeatletter
\xpatchcmd{\proof}{\topsep6\p@\@plus6\p@\relax}{}{}{}
\makeatother

\usepackage{amsthm}  
\usepackage{aliascnt} 
\usepackage[draft]{hyperref} 
\usepackage{cleveref} 

\newtheorem{theorem}{Theorem}
\crefname{theorem}{Theorem}{Theorems}

\theoremstyle{plain}

\newaliascnt{lemma}{theorem}
\newtheorem{lemma}[lemma]{Lemma}
\aliascntresetthe{lemma}
\crefname{lemma}{Lemma}{Lemmas}

\newaliascnt{definition}{theorem}
\newtheorem{definition}[definition]{Definition}
\aliascntresetthe{definition}
\crefname{definition}{Definition}{Definitions}

\newaliascnt{corollary}{theorem}

\aliascntresetthe{corollary}
\crefname{corollary}{Corollary}{Corollaries}

\newaliascnt{claim}{theorem}
\newtheorem{claim}[claim]{Claim}
\aliascntresetthe{claim}
\crefname{claim}{Claim}{Claims}

\newaliascnt{conjecture}{theorem}

\aliascntresetthe{conjecture}
\crefname{conjecture}{Conjecture}{Conjectures}

\newaliascnt{question}{theorem}

\aliascntresetthe{question}
\crefname{question}{Question}{Questions}

\newaliascnt{problem}{theorem}

\aliascntresetthe{problem}
\crefname{problem}{Problem}{Problems}

\newaliascnt{example}{theorem}

\aliascntresetthe{example}
\crefname{example}{Example}{Examples}

\newaliascnt{oquestion}{theorem}

\aliascntresetthe{oquestion}
\crefname{oquestion}{Open Question}{Open Questions}

\newaliascnt{remark}{theorem}
\newtheorem{remark}[remark]{Remark}
\aliascntresetthe{remark}
\crefname{remark}{Remark}{Remarks}

\newcommand{\bfc}{{\boldsymbol c}}

\newcommand{\bfu}{{\boldsymbol u}}
\newcommand{\bfv}{{\boldsymbol v}}

\newcommand{\bfx}{{\boldsymbol x}}
\newcommand{\bfy}{{\boldsymbol y}}

\newcommand{\bfzero}{{\mathbf 0}}
\newcommand{\bfone}{{\mathbf 1}}

\newcommand{\cC}{\mathcal{C}}
\newcommand{\cD}{\mathcal{D}}

\newcommand{\cI}{\mathcal{I}}

\DeclarePairedDelimiter\parenv{\lparen}{\rparen}
\DeclarePairedDelimiter\sparenv{\lbrack}{\rbrack}

\DeclarePairedDelimiterX\mathset[2]{\lbrace}{\rbrace}{#1 : #2}
\DeclarePairedDelimiterX\multset[2]{\lbrace\!\!\lbrace}{\rbrace\!\!\rbrace}{#1 : #2}

\newcommand{\floor}[1]{\lfloor #1 \rfloor}

\newcommand{\Fq}{\mathbb{F}_q}

\newcommand{\Ftwo}{\mathbb{F}_2}

\newcommand{\ballins}[2]{\cI_{#2}(#1)}
\newcommand{\balldel}[2]{\cD_{#2}(#1)}
\newcommand{\balldelp}[2]{\cD'_{#2}(#1)}
\newcommand{\ballinsp}[2]{\cI'_{#2}(#1)}

\newcommand{\prob}{\mathrm{Pr}}

\newcommand{\vt}[2]{\mathcal{VT}_{#1}(#2)}
\newcommand{\tauD}{\tau_{\mathrm{D}}}
\newcommand{\tauI}{\tau_{\mathrm{I}}}
\newcommand{\indicator}{\mathbbm{1}}
\newcommand{\listsub}[2]{\bar{L}_\mathrm{sub}(#1,#2)}
\newcommand{\listins}[2]{\bar{L}_\mathrm{ins}(#1,#2)}
\newcommand{\listdel}[2]{\bar{L}_\mathrm{del}(#1,#2)}

\bibliography{resources/literature}

\usepackage{fancyhdr}

\IEEEoverridecommandlockouts

\begin{document}
\title{On the Average List Size of VT Codes under\\ Two Insertions or Two Deletions \\

\thanks{The research was funded by the European Union (ERC, DNAStorage, 101045114 and EIC, DiDAX 101115134) and by the German Research Foundation (DFG) under Grant no. WA 3907/12-1. Views and opinions expressed are however those of the authors only and do not necessarily reflect those of the European Union or the European Research Council Executive Agency. Neither the European Union nor the granting authority can be held responsible for them.}
}
\author{%
  \IEEEauthorblockN{Anisha Banerjee\IEEEauthorrefmark{1},  Roni Con\IEEEauthorrefmark{2},
                    Antonia Wachter-Zeh\IEEEauthorrefmark{1},
                    and Eitan Yaakobi\IEEEauthorrefmark{3}}
  \IEEEauthorblockA{\IEEEauthorrefmark{1}%
                    School of Computation, Information and Technology, Technical University of Munich (TUM), Munich, Germany}
  \IEEEauthorblockA{\IEEEauthorrefmark{2}%
School of Electrical \& Computer Engineering, Tel Aviv University, Tel Aviv, Israel}
  \IEEEauthorblockA{\IEEEauthorrefmark{3}%
                    Department of Computer Science, Technion - Israel  Institute  of  Technology, Haifa 3200003, Israel}
    \IEEEauthorblockA{Email: anisha.banerjee@tum.de, 
       ronicon@tauex.tau.ac.il, 
      antonia.wachter-zeh@tum.de,
      yaakobi@cs.technion.ac.il}
      \\[-5ex]

}
\maketitle

\begin{abstract}
    Levenshtein's classical theorem shows that any code that can uniquely decode from $t$ deletions can also uniquely decode from $t$ insertions. This symmetry, however, does not generally hold in the list-decoding regime.
    Several prior results on the capacity and radii of list-decoding from insertions and deletions suggest that insertions are easier to handle than deletions. 
    In an effort to better understand this asymmetry between insertions and deletions under list-decoding, we study an analogous version of McEliece's \emph{average list size problem}. 
    More precisely, given a binary code $\cC$ of length $n$ that uniquely corrects $t'$ insertions (or deletions), we are interested in the average number of codewords returned by a $t$-insertion (or deletion) list decoder, where $t' < t$. 
    Focusing on the case of $t'=1$ and $t=2$, we find that under two insertions, this quantity is at most $(n+12)/4$ for any VT code of length $n$, whereas under two deletions it is at least $n/2 - 10$ when averaged over all length-$n$ VT codes. The nearly factor-of-two gap between these upper and lower bounds serves to confirm the relative ease of list-decoding of insertions compared to deletions from an average-case perspective.
\end{abstract}

\section{Introduction} \label{sec::intro}

Insertion and deletion errors and their error-correcting codes are objects of great mathematical intrigue. They find use in communication systems plagued by timing uncertainty~\cite{sklarDigitalCommunicationsFundamentals2021, kinnimentSynchronizationArbitrationDigital2007} and in storage paradigms such as DNA-based data storage~\cite{heckelCharacterizationDNAData2019, yazdiDNABasedStorageTrends2015} and racetrack memories~\cite{cheeCodingRacetrackMemories2018}. The study of these errors dates back to 1965, when Levenshtein~\cite{levenshtein1966binary} not only proposed an asymptotically-optimal construction to correct a single deletion or a single insertion, known in literature as the Varshamov-Tenengol'ts (VT) code, but also established an equivalence between a code that uniquely corrects $t$ insertions to one that uniquely corrects any combination of $t_1$ insertions and $t_2$ deletions, as long as $t_1 + t_2 \leq t$. 
In \emph{list-decoding}, the decoder is asked to output a small list of codewords that are within some radius, say $\tau$ (which is greater than the unique decoding radius), of the received word. 
Notably, insertions and deletions are not equivalent in the list-decoding setting; that is, a code that is list-decodable from $t$ insertions with a small list size does not necessarily imply it is also list-decodable from $t$ deletions with a small list size.

The problem of list-decoding was introduced by Elias~\cite{elias1957list} and Wozencraft~\cite{wozencraft1958list}, and is of great relevance due to its applications in information theory~\cite{ahlswedeChannelCapacitiesList1973, Blinovskii1986, eliasErrorcorrectingCodesList1991} and theoretical computer science~\cite{caiHardnessPermanent1999, guruswami2009unbalanced, sivakumarMembershipComparableSets1999, sudanPseudorandomGeneratorsXOR2001, lund2020list}. While list-decodable codes and algorithmic list-decoding under the Hamming metric~\cite{elias1957list, wozencraft1958list, sudan1997decoding, guruswami1998improved, sudanListDecodingAlgorithms2000, guruswamiAlgorithmicResultsList2006, guruswamiExplicitCodesAchieving2008, perniceListDecodingCapacityImplies2025, fathollahiTimeSpaceEfficientList2026} have been extensively investigated over the years, the body of work on list-decoding under insertions and deletions is fairly limited in comparison. As far as we know, the first work in the latter category is by Guruswami and Wang \cite{guruswami2017deletion}, who constructed an $\Omega(\varepsilon^3)$ rate binary code that is list decodable from $(1/2 - \varepsilon)$-fraction of deletions for any $\varepsilon > 0$. Subsequently, Haeupler, Shahrasbi and Sudan~\cite{haeupler2018synchronization} constructed list-decodable codes for insertions and deletions by leveraging an object known as synchronization strings \cite{haeupler2021synchronization}. 
Another line of research focuses on lower bounds on the list-decoding radius under insertions and deletions~\cite{wachter-zehListDecodingInsertions2018, hayashiListDecodabilityInsertions2020, liuLowerBoundListDecodability2023, yangImprovedJohnsontypeBounds2026}. 
An adjacent problem concerns determining the capacity of list-decoding under insertions and deletions~\cite{guruswami2017deletion, haeupler2018synchronization, guruswami2020optimally, haeuplerRateDistanceTradeoffsListDecodable2022, conInsertionListDecodingCapacity2026}, i.e., the largest code rate for a given fraction of errors such that the list size is upper bounded by a quantity that scales polynomially with the block length of~the code. Recently, Kiah~\textit{et al.}~\cite{kiahConstantListInsertionDeletionCodes2026} studied the same problem under the assumption that the list size $L$ is a fixed constant.

Across these prior works, we encounter multiple results that expose the aforementioned asymmetry between list-decoding from insertions and list-decoding from deletions. For instance, Haeupler, Shahrasbi and Sudan~\cite[Theorem 1.3]{haeupler2018synchronization} show that the maximum fraction of deletion errors that can be list-decoded by a code of rate $R$ is at most $1-R$, while any number of insertion errors can always be corrected, given that the underlying alphabet is sufficiently large. 
In \cite{hayashiListDecodabilityInsertions2020}, Hayashi and Yasunaga showed that any code with a large Levenshtein distance has a list-decoding radius under insertion errors that exceeds the one under deletion errors. 
A more recent result by Con, Doron and Ribeiro~\cite{conInsertionListDecodingCapacity2026} demonstrates that the capacity of list-decoding from insertions is significantly higher than the upper bound on the list-decoding capacity under deletions, in the regime where the fraction of errors goes to $0$.

This work aims to demonstrate this asymmetry from a different perspective inspired by McEliece's average list size problem~\cite{mcelieceAverageListSize2003}, which is stated below.

\textbf{Question.} \textit{Let $\cC_{t'}$ be a code that uniquely corrects $t' \geq 1$ substitutions. What is the expected number of non-causal codewords returned by a list decoder of $\cC_{t'}$ that decodes from $t$ substitutions, denoted by $\listsub{t}{\cC_{t'}}$, where $t' < t$ (the average is over all possible inputs to the algorithm)?
}

In particular, McEliece considered the Guruswami-Sudan decoder for Reed-Solomon (RS) codes, and presented exact formulas and simplified approximations (for substitutions) to the above quantity to establish that in most cases of practical interest, the Guruswami-Sudan decoder always returns a single codeword, i.e., the average list size is almost always one, or equivalently, $\listsub{t}{\cC_{t'}}$ is typically $0$. 

We study McEliece's problem in the settings of $t=2$ insertions and $t=2$ deletions, respectively, for binary codes that uniquely correct a single ($t'=1$) insertion or deletion. 
Specifically, for a code $\cC'$ of length $n$, we wish to determine $\listins{2}{\cC'}$ (resp. $\listdel{2}{\cC'}$), i.e., the average list size of a decoder of $\cC'$ under $t=2$ insertions (resp. deletions). We find that $\listins{2}{\cC'}$ is upper-bounded by $(n+12)/4$ (\Cref{th::upper-bound-2-ins}) where $\cC'$ is a VT code, while $\listdel{2}{\cC'}$, when averaged over all VT codes of length $n$, is lower-bounded by $n/2-10$ (\Cref{th::lower-bound-2-del-avg-vt}). These results not only corroborate the relative ease of list-decoding from insertions compared to deletions, but also highlight a considerable gap between the average list size for two insertions and that for two deletions.

\section{Preliminaries}

Throughout, we let $\Fq$ denote a finite field of order $q$. For any two integers $i,j$ such that $i < j$, we define $[i,j] \triangleq \{i,i+1,\ldots,j\}$. 
A \emph{run} of a vector $\bfx  \in \Fq^n$ is defined as a maximal substring of $\bfx$ composed of identical symbols. The number of runs in $\bfx$ is denoted by $\rho(\bfx)$. The $t$-deletion sphere of $\bfx \in \Fq^n$, denoted by $\balldel{\bfx}{t}$, is  the set of all sequences of length $n-t$ that can be obtained by deleting exactly $t$ symbols from $\bfx$. Similarly, the $t$-insertion sphere of $\bfx$, denoted as $\ballins{\bfx}{t}$, is the set of all length-$(n+t)$ sequences that arise from $\bfx$ following the insertion of exactly $t$ symbols from $\Fq$ into $\bfx$. 
It is well known that the insertion sphere is regular, i.e., the quantity  $|\ballins{\bfx}{t}|$ does not depend on $\bfx$ and equals $I^{(q)}_t(n) \triangleq \sum_{i=0}^{t} \binom{n+t}{i} (q-1)^i$~\cite{levenshteinEfficientReconstructionSequences2001a} (The superscript is dropped whenever $q$ is obvious from the context). We remind the reader that is not the case with deletions, i.e., $|\balldel{\bfx}{t}|$ depends on $\bfx$, and more specifically, on the number of runs in $\bfx$. We also define the indicator function $\indicator(\cdot)$, which takes an event as its argument and evaluates to $1$ if the event holds and $0$ otherwise.

Our error-correcting code of interest is the binary VT code~\cite{levenshtein1966binary}, the formal definition of which is given below.
\begin{definition} \label{def::vt-codes}
    For an integer $a \in [0,n]$, the VT code $\vt{a}{n}$ is the set of binary vectors
    \begin{equation*}
        \vt{a}{n} \triangleq \mathset*{\bfc = (c_1,\ldots,c_n)}{\sum_{i=1}^{n} i\cdot c_i \equiv a \bmod{(n+1})}.
    \end{equation*}
\end{definition}

Levenshtein showed that for any $a \in [0,n]$, $\vt{a}{n}$ can correct either a single insertion or a single deletion. It is interesting to note that the space $\Ftwo^n$ is partitioned by the union of all VT codes, i.e., $\bigcup_{a=0}^n \vt{a}{n} = \Ftwo^n$. VT codes are perfect under a single deletion~\cite{sloaneSingleDeletionCorrectingCodes2002}, implying that for all $a \in [0,n]$, each vector $\bfy \in \Ftwo^{n-1}$ decodes uniquely to a codeword in $\vt{a}{n}$.

For the purpose of our analysis in the upcoming sections, we will need
the following bounds on the cardinality of VT codes, the proofs of which can be found in Appendix~\ref{app::vt-upper-bound}.





\begin{lemma} \label{prop::vta-lower bound}\label{lem::vt-bounds}
    For any $n \geq 2$ and any $a \in [0,n]$, it holds that
    \begin{equation*}
        \frac{2^n}{n+1} \parenv*{1-2^{-\frac{2n-1}{3}}} \! < |\vt{a}{n}| < \frac{2^n}{n+1}\parenv*{1 + n\,2^{-\frac{2(n+1)}{3}}}.
    \end{equation*}
\end{lemma}

To formally introduce the average list sizes of $t$-insertion and $t$-deletion list decoders, we first define the following auxiliary quantities. For any code $\cC \subseteq \Ftwo^n$, any integer $t \geq 1$, and any $\bfu \in \Ftwo^{n+t}$ and $\bfv \in \Ftwo^{n-t}$, let
\begin{align*}
    \tauD^{(\cC)}(t,\bfu) &\triangleq |\balldel{\bfu}{t} \cap \cC| \\
    \tauI^{(\cC)}(t,\bfv) &\triangleq |\ballins{\bfv}{t} \cap \cC|.
\end{align*}

In words, $\tauD^{(\cC)}(t,\bfu)$ ($\tauI^{(\cC)}(t,\bfv)$) represents the number of~codewords of $\cC$ returned by a $t$-insertion ($t$-deletion) list decoder upon receiving some $\bfu \in \Ftwo^{n+t}$ ($\bfv \in \Ftwo^{n-t}$).

\begin{definition} \label{eq::def-li}
\label{eq::def-ld}
    The \emph{average list size of a $t$-insertion list decoder} and \emph{$t$-deletion list decoder} of a code $\cC \subseteq \Ftwo^n$ are given by
    \begin{align}
    \listins{t}{\cC} &\triangleq \frac{1}{|\cC|} \sum_{\bfc \in \cC} \frac{1}{|\ballins{\bfc}{t}|} \sum_{\bfu \in \ballins{\bfc}{t}} |\balldel{\bfu}{t} \cap \cC \setminus \{\bfc\}| \nonumber \\
    &= \frac{1}{|\cC|} \sum_{\bfc \in \cC} \frac{1}{|\ballins{\bfc}{t}|} \sum_{\bfu \in \ballins{\bfc}{t}}(\tauD^{(\cC)}(t,\bfu) - 1). \nonumber
\end{align}
and
\begin{align}
        \listdel{t}{\cC} &\triangleq \frac{1}{|\cC|} \sum_{\bfc \in \cC} \frac{1}{|\balldel{\bfc}{t}|} \sum_{\bfv \in \balldel{\bfc}{t}} |\ballins{\bfv}{t} \cap \cC \setminus \{\bfc\}| \nonumber \\
        &= \frac{1}{|\cC|} \sum_{\bfc \in \cC} \frac{1}{|\balldel{\bfc}{t}|} \sum_{\bfv \in \balldel{\bfc}{t}} (\tauI^{(\cC)}(t,\bfv) - 1).  \nonumber
    \end{align} 
\end{definition}

Essentially, $\listins{t}{\cC}$ and $\listdel{t}{\cC}$ indicate the average number of codewords, excluding the transmitted one, returned by a $t$-insertion and a $t$-deletion list decoder of $\cC$, respectively, where the transmitted codeword is drawn uniformly from $\cC$ and the received sequence is drawn uniformly from its $t$-insertion (resp.\ $t$-deletion) sphere. The remainder of this manuscript focuses on the case of $t=2$, investigating $\listins{t}{\cC}$ and $\listdel{t}{\cC}$ for the case when $\cC$ is a VT code. 

\section{List-decoding from Two Insertions}

In this section, we establish an upper bound on $\listins{2}{\cC}$, i.e., the average number of words returned by a $t=2$-insertion list decoder of $\cC$, where $\cC \subseteq \Ftwo^n$ is any VT code. To proceed along these lines, we first present some useful observations on single-deletion-correcting codes, in connection with the term $\tauD^{(\cC)}(2,\bfy)$ for some $\bfy \in \Ftwo^{n+2}$.

\begin{claim} \label{cl::sum-runs-vt-code}
    Any single-deletion-correcting code $\cC \subset \Ftwo^n$ satisfies
    \begin{equation}
        \sum_{\bfc \in \cC} \rho(\bfc) \leq 2^{n-1}. \label{eq::run-sum-cw}
    \end{equation} 
    Moreover, if $\cC$ is a VT code, then (\ref{eq::run-sum-cw}) holds with equality.
\end{claim}
\begin{IEEEproof}
    Note that for every $\bfx$, we have $|\balldel{\bfx}{1}| = \rho(\bfx)$. This is because for any $a \geq 1$, deleting a single symbol from a run $\bfzero^{a}$ (or $\bfone^{a}$) produces $\bfzero^{a-1}$ (respectively, $\bfone^{a-1}$), regardless of~which particular symbol was deleted.
    Therefore,
    \begin{align*}
        \sum_{\bfc \in \cC} \rho(\bfc) = \sum_{\bfc \in \cC} |\balldel{\bfc}{1}|
         = \sum_{\bfc \in \cC} \sum_{\bfy \in \balldel{\bfc}{1}} 1
        = \! \sum_{\bfy \in \Ftwo^{n-1}} \!\!\! |\ballins{\bfy}{1} \cap \cC|.
    \end{align*}
    Since $\cC$ corrects a single deletion, it holds that for each $\bfy \in \Ftwo^{n-1}$, $|\ballins{\bfy}{1} \cap \cC| \leq 1$ and~\eqref{eq::run-sum-cw} follows. 
    The `Moreover' part follows since VT codes are perfect single-deletion codes~\cite{sloaneSingleDeletionCorrectingCodes2002} which implies that for all $\bfy \in \Ftwo^{n-1}$, it holds that ${|\ballins{\bfy}{1} \cap \cC| = 1}$.
\end{IEEEproof}

For any $\bfy \in \Ftwo^{n+2}$, let $\balldelp{\bfy}{2} \subseteq \Ftwo^{n}$ denote the set of words obtained by deleting two symbols in the same run of $\bfy$ (note that we cannot delete a run of length $1$). 
By definition, it follows that $\balldelp{\bfy}{2} \subset \balldel{\bfy}{2}$ and that $|\balldelp{\bfy}{2}| \leq \rho(\bfy)$. Analogously, for any $\bfx \in \Ftwo^n$, let $\ballinsp{\bfx}{2}$ be the set of words obtained by inserting two identical symbols consecutively anywhere in $\bfx$. Of course, it holds that $|\ballinsp{\bfx}{2}| = I_1(n) = n+2$. These observations allow us to infer that for any code $\cC$ that corrects a single insertion or a single deletion, it holds that
\begin{align}
    \sum_{\bfy \in \Ftwo^{n+2}} |\balldelp{\bfy}{2} \cap \cC| &= \sum_{\bfy \in \Ftwo^{n+2}} \sum_{\bfc \in \cC} \indicator(\bfc \in \balldelp{\bfy}{2}) \nonumber \\    
    &= \sum_{\bfc \in \cC} \sum_{\bfy \in \Ftwo^{n+2}}  \indicator(\bfy \in \ballinsp{\bfc}{2}) \nonumber \\
    &= \sum_{\bfc \in \cC} |\ballinsp{\bfc}{2}| = (n+2)|\cC|. \label{eq::d2p-bound}
\end{align}
\begin{figure*}[t]
	\begin{center}
		\scalebox{0.7}{
			\begin{tikzpicture}[every node/.style={draw=none, fill=none, font=\large}]
				\def\hsep{4.5cm}
				
				\node[align=center] (ya) {$\bfy=010\underline{111}01$};
				
				\node[align=center, xshift=\hsep] (vadel) {$0101101$};
				
                \node[align=center, xshift=2*\hsep] (va) {$\bfv_1=010101$};

                \node[xshift=3*\hsep] (yb) {$\bfy=0\underline{1}0111\underline{0}1$};
				
				\node[align=center, xshift=4*\hsep, yshift=1.2cm] (vbdela) {$0101111$};
				
				\node[align=center, xshift=4*\hsep, yshift=-1.2cm] (vbdelb) {$0011101$};

                \node[align=center, xshift=5*\hsep] (vb) {$\bfv_2=001111$};
                
                \node[align=center, xshift=\hsep, yshift=2.05cm] (inta) {$\balldel{\bfy}{1} \cap \ballins{\bfv_1}{1}$};
                \node[align=center, xshift=4*\hsep, yshift=2.05cm] (intb) {$\balldel{\bfy}{1} \cap \ballins{\bfv_2}{1}$};

                \draw[-{Stealth[length=5pt, width=3pt]}] (ya.east) to node[pos=0.5, above=2pt, align=center]{$1$ Deletion} (vadel.west);

                \draw[-{Stealth[length=5pt, width=3pt]}] (vadel.east) to node[pos=0.5, above=2pt, align=center]{$1$ Deletion} (va.west);
                
				\draw[-{Stealth[length=5pt, width=3pt]}, bend left=14] (yb.east) to node[pos=0.5, sloped, above=2pt, align=center, font=\small]{$1$ Deletion} (vbdela.west);

                \draw[-{Stealth[length=5pt, width=3pt]}, bend left=-14] (yb.east) to node[pos=0.5, sloped, above=2pt, align=center, font=\small]{$1$ Deletion} (vbdelb.west);

                \draw[-{Stealth[length=5pt, width=3pt]}, bend left=14] (vbdela.east) to node[pos=0.5, sloped, above=2pt, align=center, font=\small]{$1$ Deletion} (vb.west);

                \draw[-{Stealth[length=5pt, width=3pt]}, bend left=-14] (vbdelb.east) to node[pos=0.5, sloped, above=2pt, align=center, font=\small]{$1$ Deletion} (vb.west);
				
				
				
			\end{tikzpicture}
            }
	\end{center}
    \vspace{-3.3mm}
	\caption{Demonstration of different deletion patterns affecting $\bfy$, leading to different values of $|\balldel{\bfy}{1} \cap \ballins{\bfv}{1}|$ in \Cref{lem::taud_upper_bound}\label{fig::intersection-sizes}. In this example, we consider $\bfy = 01011101$, and two vectors in its two-deletion sphere $\balldel{\bfy}{2}$, namely $\bfv_1=010101$ and $\bfv_2=001111$. We see that $|\balldel{\bfy}{1} \cap \ballins{\bfv_1}{1}| =1$, which is a consequence of the fact that both deletion events affect the same run, i.e., $\bfv_1 \in \balldelp{\bfy}{2}$. On the contrary, $\bfv_2$ results from a deletion in two distinct runs of~$\bfy$, i.e., $\bfv_2 \not\in \balldelp{\bfy}{2}$, leading to $|\balldel{\bfy}{1} \cap \ballins{\bfv_2}{1}| = 2$.}
    \vspace{-3.8mm}
\end{figure*}

\begin{lemma}
    For any $\bfy \in \Ftwo^{n+2}$ and a single-deletion-correcting code $\cC$ of length $n$, it holds that 
    \begin{equation*}
        \tauD^{(\cC)}(2,\bfy)\leq \frac{\rho(\bfy)+ |\balldelp{\bfy}{2} \cap \cC|}{2} \leq \rho(\bfy).\end{equation*}\label{lem::taud_upper_bound}
\end{lemma}
\begin{IEEEproof}
    Observe that $\tauD^{(\cC)}(2,\bfy)$ can be expressed as
    \begin{align}
        \tauD^{(\cC)}(2,\bfy) &= |\balldel{\bfy}{2} \cap \cC| \nonumber \\
        &= \sum_{\bfu \in \balldel{\bfy}{1}} \sum_{\bfv \in \balldel{\bfu}{1}} \frac{\indicator(\bfv \in \cC)}{|\balldel{\bfy}{1} \cap \ballins{\bfv}{1}|}. \label{eq::taud-bound}
    \end{align}
    Indeed, note that there can be more than a single deletion pattern, of the two deletions, to get from $\bfy$ to $\bfv$. 
    The quantity $|\balldel{\bfy}{1} \cap \ballins{\bfv}{1}|$ exactly counts those patterns (see Fig.~\ref{fig::intersection-sizes} for a visualization).
    More specifically, for any $\bfv \in \cC \cap \balldelp{\bfy}{2}$, we have $|\balldel{\bfy}{1} \cap \ballins{\bfv}{1}| = 1$, i.e., $\balldel{\bfy}{1} \cap \ballins{\bfv}{1}$ contains a single vector that results from causing a single deletion in the affected run of $\bfy$. On the contrary, when $\bfv \in \cC \cap \balldel{\bfy}{2} \setminus \balldelp{\bfy}{2}$, we have $|\balldel{\bfy}{1} \cap \ballins{\bfv}{1}| = 2$. Thus, each summand in (\ref{eq::taud-bound}) splits into
    \begin{align*}
        \frac{\indicator(\bfv \in \cC)}{|\balldel{\bfy}{1} \cap \ballins{\bfv}{1}|} &= \frac{\indicator(\bfv \in \cC \setminus \balldelp{\bfy}{2})}{2} + \indicator(\bfv \in \cC \cap \balldelp{\bfy}{2}) \\
        &= \frac{\indicator(\bfv \in \cC)+\indicator(\bfv \in \cC \cap \balldelp{\bfy}{2})}{2}.
    \end{align*}   
    
    This allows us to rewrite (\ref{eq::taud-bound}) as
    \begin{align*}
        &\frac{1}{2}\sum_{\bfu \in \balldel{\bfy}{1}} \sum_{\bfv \in \balldel{\bfu}{1}} \left( \indicator(\bfv \in \cC ) + \indicator(\bfv \in \cC \cap \balldelp{\bfy}{2})\right)
    \end{align*}
    Now, recall that $\cC$ is a single-deletion code, and thus, $\sum_{\bfv \in \balldel{\bfu}{1}} \indicator(\bfv \in \cC) \leq 1$ for any $\bfu \in \balldel{\bfy}{1} \subset \Ftwo^{n+1}$. Also recall that 
    $|\balldel{\bfy}{1} \cap \ballins{\bfv}{1}| = 1$ for $\bfv \in \cC \cap \balldelp{\bfy}{2}$, so 
    $\sum_{\bfu \in \balldel{\bfy}{1}} \sum_{\bfv \in \balldel{\bfu}{1}} \indicator(\bfv \in \cC \cap \balldelp{\bfy}{2}) = |\cC \cap \balldelp{\bfy}{2}|$. Thus, the first inequality in the lemma statement follows. 
     
     Furthermore, since $|\cC \cap \balldelp{\bfy}{2}| \leq |\balldelp{\bfy}{2}| \leq \rho(\bfy)$, the second inequality also readily follows.
\end{IEEEproof}

We now can deduce the following upper bound on $\listins{2}{\cC}$, where $\cC$ is any length-$n$ VT code.
\begin{theorem} \label{th::upper-bound-2-ins}
    For any $n \geq 2$, let $\cC$ be a VT code of length $n$. Then, it holds that 
    \begin{align*}
        \listins{2}{\cC} &< \frac{n+12}{4}.
    \end{align*}
\end{theorem}

\begin{IEEEproof}
By applying  \Cref{lem::taud_upper_bound} to \Cref{eq::def-li}, we get
\begin{align}
     \listins{2}{\cC} + 1 &= \frac{1}{|\cC|} \sum_{\bfc \in \cC} \frac{1}{|\ballins{\bfc}{2}|} \sum_{\bfy \in \ballins{\bfc}{2}} \tauD^{(\cC)}(2,\bfy) \nonumber \\
    &= \frac{1}{2|\cC|I_2(n)} \sum_{\bfc \in \cC}  \sum_{\bfy \in \ballins{\bfc}{2}} \!\!\!  \left( \rho(\bfy) 
    + |\balldelp{\bfy}{2} \cap \cC| \right). \label{eq::list-ins-two-terms}
\end{align}

Since inserting two symbols may increase the number of~runs by at most $4$, we have
\begin{align}
    \sum_{\bfc \in \cC}  \sum_{\bfy \in \ballins{\bfc}{2}} \rho(\bfy) &\leq \sum_{\bfc \in \cC}  \sum_{\bfy \in \ballins{\bfc}{2}} (\rho(\bfc) + 4) \nonumber \\
    &\leq I_2(n)  \parenv*{2^{n-1} +  4|\cC|}, \label{eq::sum1}
\end{align}
where the final equality follows from \Cref{cl::sum-runs-vt-code}.

We now proceed to bound the second term in (\ref{eq::list-ins-two-terms}).
\begin{align}
    &\sum_{\bfc \in \cC}  \sum_{\bfy \in \ballins{\bfc}{2}} |\balldelp{\bfy}{2} \cap \cC| \nonumber \\
    &= \sum_{\bfy \in \Ftwo^{n+2}} \sum_{\bfc \in \cC} \indicator(\bfy \in \ballins{\bfc}{2}) \cdot |\balldelp{\bfy}{2} \cap \cC| \nonumber \\
    &= \sum_{\bfy \in \Ftwo^{n+2}} |\balldelp{\bfy}{2} \cap \cC|  \cdot \sum_{\bfc \in \cC} \indicator(\bfc \in \balldel{\bfy}{2}) \nonumber \\
    &= \sum_{\bfy \in \Ftwo^{n+2}} |\balldelp{\bfy}{2} \cap \cC|  \cdot \tauD^{(\cC)}(2,\bfy). \nonumber
\end{align}

Since by \Cref{lem::taud_upper_bound}, we have $\tauD^{(\cC)}(2,\bfy) \leq \rho(\bfy) \leq n+2$, the above expression can be simplified into
\begin{align}
    \sum_{\bfc \in \cC}  \sum_{\bfy \in \ballins{\bfc}{2}} |\balldelp{\bfy}{2} \cap \cC| &\leq (n+2) \sum_{\bfy \in \Ftwo^{n+2}} |\balldelp{\bfy}{2} \cap \cC| \nonumber \\
    &= (n+2)^2 |\cC|, \label{eq::sum2}
\end{align}
where, the equality follows from (\ref{eq::d2p-bound}). Finally, we leverage (\ref{eq::sum1}) and (\ref{eq::sum2}) to bound $\listins{2}{\cC}$
\begin{align}
    \listins{2}{\cC} &< \frac{2^{n-2}}{|\cC|} + \frac{(n+2)^2}{2I_2(n)} + 1 \nonumber \\
    &< \frac{n+1}{4 \parenv*{1-2^{-\frac{2n-1}{3}}}} + 2, \label{eq::list-ins}
\end{align}
where the second inequality follows from \Cref{prop::vta-lower bound} and $I_2(n) = \frac{n^2+5n+8}{2} > \frac{(n+3)(n+2)}{2}$.

Observe that since $n \geq 2$, we have $0 < 2^{-\frac{2n-1}{3}} \leq \frac{1}{2}$, implying that $(1- 2^{-\frac{2n-1}{3}})^{-1} \leq 1 + 2 \cdot 2^{-\frac{2n-1}{3}}$. We apply this to (\ref{eq::list-ins}) below.
\begin{align*}
    \listins{2}{\cC} &<  \frac{n+1}{4}(1+2 \cdot 2^{-\frac{2n-1}{3}}) + 2 \\
    &= \frac{n+9}{4} + \frac{n+1}{2} \cdot 2^{-\frac{2n-1}{3}}.
\end{align*}

Since the latter term equals $\frac{3}{4}$ at $n=2$ and further decreases as $n$ increases, the statement of the theorem follows.
\end{IEEEproof}

\section{List-decoding from Two Deletions}

We now turn to examine $\listdel{t = 2}{\vt{a}{n}}$ for any $a \in [0,n]$. In particular, we aim to show that list-decoding from two deletions with a single-deletion (or insertion)-correcting code is more challenging than list-decoding from two insertions, by establishing a lower bound on $\listdel{t = 2}{\vt{a}{n}}$ that significantly exceeds the upper bound on $\listdel{t = 2}{\cC}$ presented in \Cref{th::upper-bound-2-ins}, where $\cC$ is any single-deletion-correcting code. We proceed along these lines by first proving a simpler lower bound on $\listdel{t = 2}{\cC}$.
\begin{lemma}
    For any $n \geq 3$ and any arbitrary $\cC \subseteq \Ftwo^n$, it holds that
    \begin{align*}
        \listdel{2}{\cC} &\geq \frac{2}{|\cC|}\sum_{\bfy \in \Ftwo^{n-2}} \frac{1}{\binom{\rho(\bfy) + 5}{2}} \binom{\tauI^{(\cC)}(2,\bfy) }{2}.
    \end{align*}
    \label{lem::1-del-code-ld}
\end{lemma}

\begin{IEEEproof}
Consider a codeword $\bfc \in \cC$ and some $\bfy \in \balldel{\bfc}{2}$. It is easy to see that
\begin{align}
    \sum_{\bfy \in \balldel{\bfc}{2}} (\tauI^{(\cC)}(2,\bfy) - 1)
    &= \sum_{\bfy \in \balldel{\bfc}{2}} (|\ballins{\bfy}{2} \cap \cC| - 1) \nonumber \\
    &= \sum_{\bfy \in \balldel{\bfc}{2}} \sum_{\bfc' \in \cC \setminus \{\bfc\}} \mathbbm{1}(\bfc' \in \ballins{\bfy}{2}) \nonumber \\
    &= \sum_{\bfc' \in \cC \setminus \{\bfc\}} \sum_{\bfy \in \balldel{\bfc}{2}} \mathbbm{1}(\bfy \in \balldel{\bfc'}{2}) \nonumber  \\
    &= \sum_{\bfc' \in \cC \setminus \{\bfc\}} |\balldel{\bfc}{2} \cap \balldel{\bfc'}{2}|. \nonumber 
\end{align}
The first equality just uses the definition of $\tauI$. The second equality is trivial while the third uses the fact that both sides count pairs $(\bfy, \bfc')$ such that $\bfc' \in C \setminus \{\bfc \}$ and $\bfy \in \balldel{\bfc}{2}$. The last equality follows by the definition of $\balldel{\bfc}{2}$.

Plugging this into $\listdel{2}{\cC}$ given in \Cref{eq::def-ld}, we get
\begin{align}
    &\listdel{2}{\cC} 
    = \frac{1}{|\cC|} \sum_{\bfc \in \cC} \frac{1}{|\balldel{\bfc}{2}|} \sum_{\bfc' \in \cC \setminus \{\bfc\}} |\balldel{\bfc}{2} \cap \balldel{\bfc'}{2}| \nonumber \\
    &= \frac{1}{|\cC|}\!\! \sum_{\substack{\{ \bfc,\bfc'\} \subseteq \cC, \bfc \neq \bfc'}}\!\!\!\!\!\!\!\! \!\!\! |\balldel{\bfc}{2} \cap \balldel{\bfc'}{2}| \!\left(\frac{1}{|\balldel{\bfc}{2}|} \!+\! \frac{1}{|\balldel{\bfc'}{2}|}\right)\!.\!\label{eq::ldel-t} 
\end{align}

Levenshtein's upper bound on $|\balldel{\cdot}{t}|$ from \cite{levenshtein2002bounds} states that $|\balldel{\bfc}{t}| \leq \binom{\rho(\bfc) + t-1}{t}$. 
Now, notice that for any $\bfy \in \balldel{\bfc}{t}$, it holds that $\rho(\bfy) \geq \rho(\bfc) - 2t$. Thus, every $\bfy \in \balldel{\bfc}{t}$ satisfies
\begin{align}
    |\balldel{\bfc}{t}| \leq \binom{\rho(\bfy) + 3t-1}{t}. \nonumber
\end{align}

We apply the preceding bound to (\ref{eq::ldel-t}) to arrive at
\begin{align}
    &\listdel{2}{\cC} \nonumber \\
    &\geq \frac{1}{|\cC|} \sum_{\substack{\{ \bfc,\bfc'\} \subseteq \cC, \bfc \neq \bfc'}}  \sum_{\bfy \in \Ftwo^{n-2}} \indicator(\bfy \in \balldel{\bfc}{2} \cap \balldel{\bfc'}{2}) \frac{2}{\binom{\rho(\bfy) + 5}{2}} \nonumber \\
    &= \frac{1}{|\cC|}\sum_{{\bfy \in \Ftwo^{n-2}}} \frac{2}{\binom{\rho(\bfy) + 5}{2}} \sum_{\substack{\{ \bfc,\bfc' \} \subseteq \cC , \bfc \neq \bfc'}} \indicator(\bfy \in \balldel{\bfc}{2} \cap \balldel{\bfc'}{2}) \nonumber \\
    &= \frac{2}{|\cC|} \sum_{\bfy \in \Ftwo^{n-2}} \frac{1}{\binom{\rho(\bfy) + 5}{2}} \binom{|\ballins{\bfy}{2} \cap \cC|}{2}, \label{eq::list-2-del-c-lower}
\end{align}
where the final equality follows since $\sum_{\{ \bfc,\bfc' \} \subseteq \cC,\bfc \neq \bfc'} \indicator(\bfy \in \balldel{\bfc}{2} \cap \balldel{\bfc'}{2})$ counts how many pairs of codewords have $\bfy$ as a subsequence, which is exactly $\binom{|\ballins{\bfy}{2} \cap \cC|}{2}$.
\end{IEEEproof}

The preceding lemma suggests that establishing a lower bound on $\listdel{2}{\vt{a}{n}}$ for any $a \in [0,n]$ requires us to lower bound the summation in (\ref{eq::list-2-del-c-lower}), which seems to be a challenging task to accomplish for a specific instance of $a$. Instead, we seek to derive a lower bound on $\listdel{2}{\vt{a}{n}}$ averaged over all $a \in [0,n]$, as follows.

\begin{theorem}  \label{th::lower-bound-2-del-avg-vt}
    For any $n\geq 3$, it holds that 
    \begin{align*}
        \frac{1}{n+1} \sum_{a=0}^{n} \listdel{2}{\vt{a}{n}} > \frac{n}{2} - 10. 
    \end{align*}    
\end{theorem}
\begin{IEEEproof}
From \Cref{lem::1-del-code-ld} and \Cref{prop::vta-lower bound}, it directly follows that
\begin{align*}
    &\frac{1}{n+1} \sum_{a=0}^{n} \listdel{2}{\vt{a}{n}} \\
    &\geq \! \frac{2}{2^n (1+2^{-\frac{2(n+1)}{3}})} \! \sum_{\bfy \in \Ftwo^{n-2}} \!\! \frac{1}{\binom{\rho(\bfy)+5}{2}} \!\sum_{a=0}^{n} \! \binom{\tauI^{(\vt{a}{n})}(2,\bfy)}{2}.
\end{align*}

Since $\bigcup_{a \in [0,n]} \vt{a}{n} = \Ftwo^n$ and all VT codes are disjoint, we have $\sum_{a=0}^{n}|\ballins{\bfy}{2} \cap \vt{a}{n}| = |\ballins{\bfy}{2}| = I_2(n-2) = (n^2+n+2)/2 > n(n+1)/2$. Note that the convexity of the binomial coefficient allows us to use Jensen's inequality~\cite{mitzenmacherProbabilityComputingRandomized2005} to write
\begin{align}
    \sum_{a=0}^{n} \binom{|\ballins{\bfy}{2} \cap \vt{a}{n}|}{2} &\geq (n+1) \binom{\frac{I_2(n-2)}{(n+1)}}{2} \nonumber \\
    &>(n+1) \binom{n/2}{2} \nonumber\\
    &= \frac{(n+1)n(n-2)}{8}. \label{eq::bound-i2-vta}
\end{align}
Furthermore, observe that
\begin{align}
    \sum_{\bfy \in \Ftwo^{n-2}} \frac{1}{\binom{\rho(\bfy)+5}{2}} &> \sum_{r=1}^{n-2} 2\binom{n-3}{r-1} \frac{2}{(r+5)^2} \nonumber \\
    &= 4\sum_{r=0}^{n-3} \binom{n-3}{r} \frac{1}{(r+6)^2}. \label{eq::binom-sum}
\end{align}

To bound the preceding summation, we note from \Cref{lem::binom-sum-jensen} in Appendix~\ref{app::sum-binom-jensen} that
\begin{align}
    \sum_{r=0}^{n-3} \binom{n-3}{r} \frac{1}{(r+6)^{2}} &\geq \frac{2^{n-3}}{(\frac{n-3}{2}+6)^2}.\label{eq::binom-sum2}
\end{align}

Upon combining (\ref{eq::bound-i2-vta}), (\ref{eq::binom-sum}) and (\ref{eq::binom-sum2}), we are led to
\begin{align}
    \frac{1}{n+1} \!\! \sum_{a=0}^{n} \! \listdel{2}{\vt{a}{n}} \!
    &> \! \frac{(n+1)n(n-2)}{8(1+2^{-\frac{2(n+1)}{3}})(\frac{n-3}{2}+6)^2} \nonumber \\
    \!&= \!\! \frac{(n+1)n(n-2)}{2(1+2^{-\frac{2(n+1)}{3}})(n+9)^2}.\!\! \label{eq::lower-bound-ld}
\end{align}

To simplify this bound further, we let $g(n) \triangleq \frac{n(n+1)(n-2)}{2(n+9)^2}$ and $u(n) \triangleq 2^{-\frac{2(n+1)}{3}}$. We note that
\begin{align*}
    \frac{n}{2} - g(n) = \frac{n(19n+83)}{2(n+9)^2} < \frac{19}{2}.
\end{align*}
Moreover $g(n) < n/2$, so $\frac{g(n)}{1+u(n)} > g(n) - g(n)u(n) > g(n) - \frac{n}{2}2^{-\frac{2(n+1)}{3}}$. Since $\frac{n}{2}2^{-\frac{2(n+1)}{3}}$ is decreasing on $n \geq 3$, it achieves its maximum value $\frac{3}{2}2^{-8/3} < \frac{1}{4}$ at $n=3$, thus allowing us to write $\frac{g(n)}{1+u(n)} > g(n) - 1/4$. Finally, we apply this to (\ref{eq::lower-bound-ld}) to arrive at
\begin{align*}
    \frac{1}{n+1} \sum_{a=0}^{n} \listdel{2}{\vt{a}{n}} > \frac{n}{2} - \frac{19}{2} - \frac{1}{4} > \frac{n}{2} - 10 . \qedhere
\end{align*}
\end{IEEEproof}

\begin{remark}
    \Cref{th::lower-bound-2-del-avg-vt} aligns with the result derived by an intuitive argument analogous to the one offered by McEliece in \cite{mcelieceAverageListSize2003} for the substitution-only case. This approach interprets $\listdel{2}{\vt{a}{n}}$ as approximately equal to the average number of codewords in a $2$-insertion sphere of a vector drawn uniformly at random from $\Ftwo^{n-2}$. That is, 
\begin{align*}
    \listdel{2}{\vt{a}{n}} &\approx \frac{|\vt{a}{n}|}{2^n} I_2(n-2) \\
    &\approx \frac{1}{n+1} \sum_{i=0}^{2} \binom{n-2+2}{i} = \frac{n}{2} + \frac{1}{n+1}.
\end{align*}
\end{remark}

\appendix

\subsection{Bounds on Cardinalities of VT Codes} \label{app::vt-upper-bound} \label{app::vta-lower}

We know from \cite[Corollary 2.3]{sloaneSingleDeletionCorrectingCodes2002} that for any $a \in [0,n]$,
\begin{equation}
    |\vt{1}{n}| \leq |\vt{a}{n}| \leq |\vt{0}{n}|. \label{eq::vta-bounds}
\end{equation}

Thus, to prove \Cref{lem::vt-bounds}, it suffices to upper-bound $|\vt{0}{n}|$ and to lower-bound $|\vt{1}{n}|$, as has been done in \Cref{lem::vt0-code-upper-bound} and \Cref{prop::vt1-lower-bound}, respectively. Combining these bounds with (\ref{eq::vta-bounds}) concludes the proof.

\begin{lemma} \label{lem::vt0-code-upper-bound}
    For any $n \geq 2$, it holds that
    \begin{equation*}
        |\vt{0}{n}| < \frac{2^n}{n+1} \parenv*{1+2^{-2(n+1)/3}}.
    \end{equation*}
\end{lemma}
\begin{IEEEproof}
 \cite[Corollary~2.3]{sloaneSingleDeletionCorrectingCodes2002} states that
    \begin{align}
        |\vt{0}{n}| = \frac{1}{2(n+1)} \sum_{\substack{d|n+1 \\ d \text{ odd}}} \phi(d)2^{(n+1)/d}, \label{eq::vt0-1}
    \end{align}
    where $\phi(d)$ (known as the Euler totient function \cite{EulersTotientFunction2026}) counts the positive integers up to the integer $d$ that are relatively prime to $d$. To simplify the summation in the preceding expression, we first use $\phi(1) = 1$ to arrive at
    \begin{align}
        \sum_{\substack{d|n+1 \\ d \text{ odd}}} \phi(d)2^{(n+1)/d}\leq 2^{n+1} + 2^{(n+1)/3}\sum_{\substack{d|n+1 \\ d \text{ odd}, d \geq 3}} \!\!\!\! \phi(d), \label{eq::vt0-phi-d}
    \end{align}
    
    Next, we apply Gauss' observation~\cite{gauss1801disquisitiones} that ${\sum_{d|n} \phi(d) =n}$. This leads us to
    \begin{align}
        \sum_{\substack{d|n+1 \\ d \text{ odd}, d \geq 3}} \phi(d) \leq \sum_{d|n+1} \phi(d) - \phi(1) = n. \label{eq::vt0-3}
    \end{align}
        
    Upon incorporating (\ref{eq::vt0-phi-d}) and (\ref{eq::vt0-3}) into (\ref{eq::vt0-1}), we obtain
    \begin{align*}
        |\vt{0}{n}| &\leq \frac{2^{n}}{n+1}  + 2^{(n-2)/3} \frac{n}{n+1}.
    \end{align*}
    Since $\tfrac{n}{n+1} < 1$, the statement of the lemma follows.
\end{IEEEproof}


\begin{lemma} \label{prop::vt1-lower-bound}
    For any $n \geq 2$, it holds that
    \begin{equation*}
        |\vt{1}{n}| > \frac{2^n}{n+1} \parenv*{1-2^{-\frac{2n-1}{3}}}.
    \end{equation*}
\end{lemma}
\begin{IEEEproof}
    We know from \cite[Corollary 2.3]{sloaneSingleDeletionCorrectingCodes2002} that
    \begin{align}
        |\vt{1}{n}| = \frac{1}{2(n+1)} \sum_{\substack{d|n+1 \\ d \text{ odd}}} \mu(d)2^{(n+1)/d}, \label{eq::vt1-size}
    \end{align}
    where $\mu(d)$ is called the Möbius function~\cite{mobius1832} and is defined as
    \begin{align*}
        \mu(d) = \begin{cases}
            1 \!\! & \!\! \text{if } d=1, \\
            (-1)^k \!\! & \!\! \text{if $d$ is the product of $k$ distinct primes,} \\
            0 \!\! & \!\! \text{if $d$ is divisible by a square greater than $1$.}            
        \end{cases}
    \end{align*}

    Since $\mu(d)$ can only assume one of the values in $\{-1,0,1\}$, it is easy to see that the summation in (\ref{eq::vt1-size}) is subject to the following lower bound.
    \begin{align}
        \sum_{\substack{d|n+1 \\ d \text{ odd}}} \mu(d)2^{(n+1)/d} &\geq  2^{n+1} - \sum_{\substack{d|n+1 \\ d \text{ odd}, d \geq 3}} 2^{(n+1)/d}. \label{eq::mu-d-sum}
    \end{align}

    Note that in the summation on the right hand side, each term is of the form $2^a$ where $a \in [1, \floor{\frac{n+1}{3}}]$. This allows us to simplify the bound in (\ref{eq::mu-d-sum}) as
    \begin{align*}
        \sum_{\substack{d|n+1 \\ d \text{ odd}}} \mu(d)2^{(n+1)/d} &\geq  2^{n+1} - (2 + 2^2 + \ldots + 2^{\floor{\frac{n+1}{3}}}) \\
        &= 2^{n+1} - (2^{\floor{\frac{n+1}{3}}+1} - 2) \\
        &> 2^{n+1} - 2^{\frac{n+1}{3}+1} = 2^{n+1} \parenv*{\!1-2^{-\frac{2n-1}{3}}\!}.
    \end{align*}

    Upon applying this to (\ref{eq::vt1-size}), we are led to the statement of~the lemma.    
\end{IEEEproof}

\subsection{Auxiliary Lemma for Proof of \Cref{th::lower-bound-2-del-avg-vt}} \label{app::sum-binom-jensen}

\begin{lemma} \label{lem::binom-sum-jensen}
    For any integer $n \geq 1$, it holds that
    \begin{align*}
        \sum_{r=0}^{n} \binom{n}{r} \frac{1}{(r+6)^{2}} \geq \frac{2^{n}}{\left(\frac{n}{2}+6\right)^2}.
    \end{align*}
\end{lemma}
\begin{IEEEproof}
    Assume a random variable $X \sim \mathrm{Bin}(n,1/2)$ so that for any $r \in [0,n]$, we have $\prob(X = r) = 2^{-n} \binom{n}{r}$. Since $(r+6)^{-2}$ is a convex function of $r$, we can apply Jensen's inequality~\cite{mitzenmacherProbabilityComputingRandomized2005} to write
    \begin{align}
        \sum_{r=0}^{n} \binom{n}{r} \frac{1}{(r+6)^{2}} &= 2^{n} \mathbb{E}\sparenv*{\frac{1}{(X+6)^2}}\nonumber \\
        &\geq \frac{2^{n}}{(\mathbb{E}[X]+6)^2} = \frac{2^{n}}{(\frac{n}{2}+6)^2}.\nonumber \qedhere
    \end{align}
\end{IEEEproof}



\section*{AI statement}
No AI system was used to generate any part of this manuscript. Claude Opus 5.5 was used to find some of the references. 

\balance
\printbibliography

\end{document}